\documentclass[aps,pra,onecolumn,notitlepage,superscriptaddress,showpacs,nofootinbib]{revtex4-1}
\usepackage{amsmath}
\usepackage{amssymb}
\usepackage{amsfonts}
\usepackage{enumerate}
\usepackage{mathrsfs}
\usepackage{graphicx}
\usepackage{epstopdf}
\usepackage{enumerate}
\usepackage{changepage}
\usepackage{bm}

\usepackage{lipsum}

\newtheorem{theorem}{Theorem}
\newtheorem{definition}{Definition}

\newtheorem{lemma}{Lemma}
\newtheorem{proposition}{Proposition}
\newtheorem{conjecture}{Conjecture}
\newtheorem{example}{Example}
\newtheorem{corollary}{Corollary}

\def\bcj{\begin{conjecture}}
	\def\ecj{\end{conjecture}}
\def\bcr{\begin{corollary}}
	\def\ecr{\end{corollary}}
\def\bd{\begin{definition}}
	\def\ed{\end{definition}}
\def\bea{\begin{eqnarray}}
	\def\eea{\end{eqnarray}}
\def\bem{\begin{enumerate}}
	\def\eem{\end{enumerate}}
\def\bex{\begin{example}}
	\def\eex{\end{example}}
\def\bim{\begin{itemize}}
	\def\eim{\end{itemize}}
\def\bl{\begin{lemma}}
	\def\el{\end{lemma}}
\def\bma{\begin{bmatrix}}
	\def\ema{\end{bmatrix}}
\def\bpf{\begin{proof}}
	\def\epf{\end{proof}}
\def\bpp{\begin{proposition}}
	\def\epp{\end{proposition}}
\def\bqu{\begin{question}}
	\def\equ{\end{question}}
\def\br{\begin{remark}}
	\def\er{\end{remark}}
\def\bt{\begin{theorem}}
	\def\et{\end{theorem}}

\def\squareforqed{\hbox{\rlap{$\sqcap$}$\sqcup$}}
\def\qed{\ifmmode\squareforqed\else{\unskip\nobreak\hfil
		\penalty50\hskip1em\null\nobreak\hfil\squareforqed
		\parfillskip=0pt\finalhyphendemerits=0\endgraf}\fi}
\def\endenv{\ifmmode\;\else{\unskip\nobreak\hfil
		\penalty50\hskip1em\null\nobreak\hfil\;
		\parfillskip=0pt\finalhyphendemerits=0\endgraf}\fi}
\newenvironment{proof}{\noindent \textbf{{Proof.~} }}{\qed}
\def\Dbar{\leavevmode\lower.6ex\hbox to 0pt
	{\hskip-.23ex\accent"16\hss}D}
\makeatletter
\def\url@leostyle{%
	\@ifundefined{selectfont}{\def\UrlFont{\sf}}{\def\UrlFont{\small\ttfamily}}}
\makeatother
\def\bcj{\begin{conjecture}}
	\def\ecj{\end{conjecture}}
\def\bcr{\begin{corollary}}
	\def\ecr{\end{corollary}}
\def\bd{\begin{definition}}
	\def\ed{\end{definition}}
\def\bea{\begin{eqnarray}}
	\def\eea{\end{eqnarray}}
\def\bem{\begin{enumerate}}
	\def\eem{\end{enumerate}}
\def\bex{\begin{example}}
	\def\eex{\end{example}}
\def\bim{\begin{itemize}}
	\def\eim{\end{itemize}}
\def\bl{\begin{lemma}}
	\def\el{\end{lemma}}
\def\bpf{\begin{proof}}
	\def\epf{\end{proof}}
\def\bpp{\begin{proposition}}
	\def\epp{\end{proposition}}
\def\bqu{\begin{question}}
	\def\equ{\end{question}}
\def\br{\begin{remark}}
	\def\er{\end{remark}}
\def\bt{\begin{theorem}}
	\def\et{\end{theorem}}

\def\btb{\begin{tabular}}
	\def\etb{\end{tabular}}

	\newcommand{\nc}{\newcommand}
	
	\nc{\bbA}{\mathbb{A}} \nc{\bbB}{\mathbb{B}} \nc{\bbC}{\mathbb{C}}
	\nc{\bbD}{\mathbb{D}} \nc{\bbE}{\mathbb{E}} \nc{\bbF}{\mathbb{F}}
	\nc{\bbG}{\mathbb{G}} \nc{\bbH}{\mathbb{H}} \nc{\bbI}{\mathbb{I}}
	\nc{\bbJ}{\mathbb{J}} \nc{\bbK}{\mathbb{K}} \nc{\bbL}{\mathbb{L}}
	\nc{\bbM}{\mathbb{M}} \nc{\bbN}{\mathbb{N}} \nc{\bbO}{\mathbb{O}}
	\nc{\bbP}{\mathbb{P}} \nc{\bbQ}{\mathbb{Q}} \nc{\bbR}{\mathbb{R}}
	\nc{\bbS}{\mathbb{S}} \nc{\bbT}{\mathbb{T}} \nc{\bbU}{\mathbb{U}}
	\nc{\bbV}{\mathbb{V}} \nc{\bbW}{\mathbb{W}} \nc{\bbX}{\mathbb{X}}
	\nc{\bbZ}{\mathbb{Z}}
	
	\nc{\bA}{{\bf A}} \nc{\bB}{{\bf B}} \nc{\bC}{{\bf C}}
	\nc{\bD}{{\bf D}} \nc{\bE}{{\bf E}} \nc{\bF}{{\bf F}}
	\nc{\bG}{{\bf G}} \nc{\bH}{{\bf H}} \nc{\bI}{{\bf I}}
	\nc{\bJ}{{\bf J}} \nc{\bK}{{\bf K}} \nc{\bL}{{\bf L}}
	\nc{\bM}{{\bf M}} \nc{\bN}{{\bf N}} \nc{\bO}{{\bf O}}
	\nc{\bP}{{\bf P}} \nc{\bQ}{{\bf Q}} \nc{\bR}{{\bf R}}
	\nc{\bS}{{\bf S}} \nc{\bT}{{\bf T}} \nc{\bU}{{\bf U}}
	\nc{\bV}{{\bf V}} \nc{\bW}{{\bf W}} \nc{\bX}{{\bf X}}
	\nc{\ba}{{\bf a}} \nc{\be}{{\bf e}} \nc{\bu}{{\bf u}}
	\nc{\brr}{{\bf r}}
	
	\nc{\cA}{{\cal A}} \nc{\cB}{{\cal B}} \nc{\cC}{{\cal C}}
	\nc{\cD}{{\cal D}} \nc{\cE}{{\cal E}} \nc{\cF}{{\cal F}}
	\nc{\cG}{{\cal G}} \nc{\cH}{{\cal H}} \nc{\cI}{{\cal I}}
	\nc{\cJ}{{\cal J}} \nc{\cK}{{\cal K}} \nc{\cL}{{\cal L}}
	\nc{\cM}{{\cal M}} \nc{\cN}{{\cal N}} \nc{\cO}{{\cal O}}
	\nc{\cP}{{\cal P}} \nc{\cQ}{{\cal Q}} \nc{\cR}{{\cal R}}
	\nc{\cS}{{\cal S}} \nc{\cT}{{\cal T}} \nc{\cU}{{\cal U}}
	\nc{\cV}{{\cal V}} \nc{\cW}{{\cal W}} \nc{\cX}{{\cal X}}
	\nc{\cZ}{{\cal Z}}
	
	\nc{\hA}{{\hat{A}}} \nc{\hB}{{\hat{B}}} \nc{\hC}{{\hat{C}}}
	\nc{\hD}{{\hat{D}}} \nc{\hE}{{\hat{E}}} \nc{\hF}{{\hat{F}}}
	\nc{\hG}{{\hat{G}}} \nc{\hH}{{\hat{H}}} \nc{\hI}{{\hat{I}}}
	\nc{\hJ}{{\hat{J}}} \nc{\hK}{{\hat{K}}} \nc{\hL}{{\hat{L}}}
	\nc{\hM}{{\hat{M}}} \nc{\hN}{{\hat{N}}} \nc{\hO}{{\hat{O}}}
	\nc{\hP}{{\hat{P}}} \nc{\hR}{{\hat{R}}} \nc{\hS}{{\hat{S}}}
	\nc{\hT}{{\hat{T}}} \nc{\hU}{{\hat{U}}} \nc{\hV}{{\hat{V}}}
	\nc{\hW}{{\hat{W}}} \nc{\hX}{{\hat{X}}} \nc{\hZ}{{\hat{Z}}}
	
	\nc{\hn}{{\hat{n}}}

	\usepackage[
	colorlinks,
	linkcolor = blue,
	citecolor = blue,
	urlcolor = blue]{hyperref}
	\def \qed {\hfill \vrule height7pt width 7pt depth 0pt}
	
	\newcounter{lastnote}

\begin{document}
		
		\title{ Small set of orthogonal product states with nonlocality }

		\author{Yan-Ling Wang}
		\email{wangylmath@yahoo.com}
		\affiliation{ School of Computer Science and Technology, Dongguan University of Technology, Dongguan, 523808, China}
		
		\author{Wei Chen}	
				\email{2016025@dgut.edu.cn}
		\affiliation{ School of Computer Science and Technology, Dongguan University of Technology, Dongguan, 523808, China}

	\author{Mao-Sheng Li}
\email{li.maosheng.math@gmail.com}
\affiliation{ School of Mathematics,
	South China University of Technology, Guangzhou
	510641,  China}

		%\date{\today}

		\begin{abstract}
			A set of orthogonal states in multipartite systems is called to be  locally stable if to preserving the orthogonality of the states, only trivial local measurement can  be	performed from each partite. Locally stable set of states are always locally indistinguishable  yielding a form of nonlocality which is different from the Bell type nonlocality. In this work, we study the locally stable set of product states with small size.
			First, we give a lower bound on the size of  locally stable set of product states. It  is well known that  unextendible product basis (UPB) is locally indistinguishable. But we find that some of them are not locally stable.  On the other hand, there exists some small subset of minimum size UPB that are also locally stable which implies the nonlocality of such  UPBs are stronger than other form.   
			
		\end{abstract}

		\maketitle
		\section{Introduction} Quantum nonlocality plays an important role in quantum theory.   The most striking one is the Bell nonlocality which could arise only in the appearence of entanglement\cite{nils}.  Bennett   et al.'s   \cite{Ben99}   pointed out that there is another kind of nonlocality (could be deduced without entanglement) which is different the Bell type nonlocality (must with entanglement) based on the concept of quantum states discrimination.  The quantum states discrimination task is to identify the randomly chosen state from a  known set.   It is well known that a set of states can be perfectly distinguished if and only if the known set is mutually orthogonal.  As    quantum states are usually  distributed in composite systems, therefore it is much more meaningful to restrict our measurement to   local operations and classical
		communication (LOCC).  We call such problem as  a local discrimination task.  Given an orthogonal set of multipartite quantum states,   if the correspoding local discrimination task can be   accomplished perfectly, we say that the set is \emph{locally distinguishable},
		otherwise, \emph{locally indistinguishable}.   Bennett et al. \cite{Ben99} presented the first set 
		of nine orthogonal product states in $\mathbb{C}^3\otimes  \mathbb{C}^3$ that are locally indistinguishable  which implies some global measurement will yield more information than those obtained from local measurement.    Therefore, they named such phenomenon as quantum nonlocality without entanglement.   Moreover, the results of  local  indistinguishability of  quantum states have  been practically applied in quantum cryptography primitives such as   data hiding \cite{Terhal01,DiVincenzo02} and secret sharing \cite{Markham08,Rahaman15,WangJ17}.

		Since  Bennett et al. \cite{Ben99} proposed such kind of nonlocality, lots of related research     has been studied extensively (see Refs. \cite{Gho01,Wal00,Wal02,Fan04,Nathanson05,Cohen07,Bandyopadhyay11,Li15,Fan07,Yu12,Cos13,Yu115,Wang19,Xiong19,Li20,Banik21,Ran04,Hor03,Ben99b,DiVincenzo03,Zhang14,Zhang15,Zhang16,Xu16b,Xu16m,Xu17QIP,Zhang16b,Wang15,Feng09,		Yang13,Zhang17,Halder18,Li18,Halder1909,Xu20b,Rout1909,Halder20c,Halder19,Xu20a,Bhunia21,Zhang21e,Zuo21}  for an
		incomplete list). 
		Most of the results are based on the method proposed by Walgate and   Hardy \cite{Wal02}.  The idea is based on the fact every measurement in a perfectly local distinguishing protocol must be orthogonality preserving.  Given a set of multipartite quantum states, if one can show that  only   trivial local orthogonal preserving measurement could be performed by  each partite, then one can conclude that the set is locally indistinguishable.  Using this method, there are a lots of   studies extend Bennett et al's results to general multipartite systems (see Refs. \cite{Zhang14,Zhang15,Wang15,Zhang16,Xu16b,Xu16m,Zhang16b,Xu17QIP,Zhang17,Halder18,Li18,Halder1909,Rout1909,Xu20b,Xu20a,Halder20c}). Recently, there are several kinds of constructions of locally indistinguishable sets of product states in multipartite systems. The cardinality of such sets are always different. In general, the smaller of  the cardinality, the better of the constructing set. Therefore, it is interesting to determine the smallest locally indistinguishable set for a given system.
		However, there is a trivial answer  as we can always embedding sets in small dimension to large dimension and  a few partite to  more partite. Therefore, in order to avoid this embarrassment we replace this kind of nonlocality by a strong one (locally stable set) which is recently introduced in Ref. \cite{Shi_Li21,Li21}. In this paper, we will study the quantum nonlocality of set of product states based on the concept of locally stable set.

		An orthogonal set of states is said to be locally stable (See Ref. \cite{Li21}) if  every orthogonality preserving measurement  for each subsystems is trivial. Therefore, locally stable sets are always locally indistinguishale sets. But the converse statement is true in general. It is well known that every unextendible product basis (UPB) (a set of incomplete
		orthonormal product states whose complementary space has no product state, See Ref. \cite{DiVincenzo03,AL01,CJ15,J14} ) is locally indistinguishable.   So it is interesting to consider whether UPBs are also locally stable or not. Moreover, the UPB of  smallest size   for a given system seems to be the smallest nonlocality set of product states  for a given system. However, we will show that this intuitive understanding is incorrect. In fact, there exists a much small subset of such UPB which is still locally stable.

		The rest of this article is organized as follows. In Sec. \ref{sec:Pre}, we   introduce a special form of quantum nonlocality called local stable set and review some known results about it. In Sec. \ref{sec:smallset}, we study locally stable set of product states including both upper bound and lower bound on the size of such set.  Moreover, we pay attention to the $n$ qubit systems and $n$ qutrit systems and study the local stableness of some UPBs.       Finally, we draw a conclusion and present some interesting problems in   the last section \ref{sec:Conclusion}.

		\vskip 12pt	
		
		\section{Preliminaries}\label{sec:Pre}

		For any positive integer $d\geq 2$, we denote $\bbZ_{d}:=\{0,1,\cdots,d-1\}$ as the cyclic group with $d$ elements  and we denote $[d]$ as the set $\{1,2,\cdots,d\}$. 	Let $\cH$ be an $d$ dimensional Hilbert space. We always assume that $\{|0\rangle, |1\rangle, \cdots,|d-1\rangle\}$ is the computational basis of $\cH$. A positive operator-valued measure (POVM) on $\mathcal{H}$ is a set of  positive semidefinite operators $\{E_x\}_{x\in\cX}$ such that $\sum_{x\in\cX} E_x= \mathbb{I}_{\mathcal{H}}$ where $\mathbb{I}_{\mathcal{H}}$  is the identity operator on   ${\mathcal{H}}$.
		The measurement  $\{E_x\}_{x\in\cX}$  is  called a  trivial measurement if   $E_x\propto \mathbb{I}_\mathcal{H}$.

		As nonorthogonal set of states can not be distignuished perfectly even   global measurements are allowed, the measurements  in the perfect LOCC distinguishing protocol are always  orthogonality preserving local
		measurements.			 To prove a set of orthogonal states are locally indistinguishable, an usual method is to show that  every  orthogonality preserving local
		measurement is a trivial measurement.
		Motivated by this, the authors in Ref. \cite{Li21}   introduced  the following concept.
		
		\begin{definition}[Locally stable]\label{def: locally stable}
			An orthogonal set of pure states in multipartite quantum system is said to be locally stable if only trivial orthogonality preserving measurement can be performed for each local subsystems.
		\end{definition}

		Let $\mathcal{H}_{A}\otimes\mathcal{H}_B$ be a composed bipartite systems whose local dimensions  are $d_A$ and $d_B$ respectively and suppose $\{|i\rangle_{A}| i\in\mathbb{Z}_{d_A}\}$ and $\{|j\rangle_{B}| j\in\mathbb{Z}_{d_B}\}$ be the computational bases of  subsystems $A$ and $B$ respectively. Given an orthogonal set $\cS=\{|\psi_k\rangle\}_{k=1}^N$ of pure states in  $\mathcal{H}_{A}\otimes\mathcal{H}_B$. The authors in Ref. \cite{Li21} gave a sufficient and necessary condition for the set $\mathcal{S}$ to be locally stable.  Now we give a brief review  of  such a condition.  
		
		The statement ``Bob could only start with trivial measurement" is equivalent to that $$  
		\langle \psi_k|\mathbb{I}_{A} \otimes E |\psi_l\rangle=0, \text{ for } k\neq  l $$
		implies $E\propto \mathbb{I}_B.$ It was shown that this is equivalent to 
		\begin{equation}\label{eq:dimeq} \text{dim}_\mathbb{C}[\mathcal{D}_{A|B}(\mathcal{S})]=d_B^2-1.
		\end{equation}
		Here $\mathcal{D}_{A|B}(\mathcal{S})$ is the linear subspace of $\text{Mat}_{d_B\times d_B}(\mathbb{C})$ spanned by 
		$$	\sum_{i\in \mathbb{Z}_{d_A}}   |\psi_{k,i}\rangle_{B}\langle \psi_{l,i} |, 1\leq k\neq l \leq N$$
		where   $|\psi_k\rangle$ and  $|\psi_l\rangle$ are  written  as the following form  
		$$\begin{array}{c}
			|\psi_k\rangle=\displaystyle\sum_{ i\in \mathbb{Z}_{d_A}}|i\rangle_{A}|\psi_{k,i}\rangle_{B},  \ \ 
			|\psi_l\rangle=\displaystyle\sum_{ i\in \mathbb{Z}_{d_A}}|i\rangle_{A}|\psi_{l,i}\rangle_{B}. 
		\end{array}$$
		
		In this paper, we mainly consider the sets of product states. Given $\cS=\{|\psi_k\rangle=|\theta_k\rangle_A|\xi_k\rangle_B\}_{k=1}^N$ of pure  product states in  $\mathcal{H}_{A}\otimes\mathcal{H}_B.$ For $k\neq l$, if  $|\theta_k\rangle_A=\sum_{i\in \mathbb{Z}_{d_A}} a_{k,i} |i\rangle_A$ and  $|\theta_l\rangle_A=\sum_{i\in \mathbb{Z}_{d_A}} a_{l,i} |i\rangle_A$,  then 
		$$\begin{array}{c}
			|\psi_k\rangle=\displaystyle\sum_{ i\in \mathbb{Z}_{d_A}}|i\rangle_{A} a_{k,i}|\xi_k\rangle_{B},  \ \ 
			|\psi_l\rangle=\displaystyle\sum_{ i\in \mathbb{Z}_{d_A}}|i\rangle_{A} a_{l,i}|\xi_l\rangle_{B}. 
		\end{array}$$
		That is,   $|\psi_{k,i}\rangle_{B}=a_{k,i} |\xi_k\rangle_B$ and $|\psi_{l,i}\rangle_{B}=a_{l,i} |\xi_l\rangle_B$ for $i\in\mathbb{Z}_{d_A}.$ Therefore, 
		$$\sum_{i\in \mathbb{Z}_{d_A}}   |\psi_{k,i}\rangle_{B}\langle \psi_{l,i} | = \left(\sum_{i\in \mathbb{Z}_{d_A}} a_{k,i} \overline{a_{l,i}} \right) |\xi_k\rangle_B\langle \xi_l|.$$  
		Note that  $\langle \theta_l|\theta_k\rangle_A =\sum_{i\in \mathbb{Z}_{d_A}} a_{k,i} \overline{a_{l,i}}$. Therefore, if  $\langle \theta_l|\theta_k\rangle_A=0,$ we have $\sum_{i\in \mathbb{Z}_{d_A}}   |\psi_{k,i}\rangle_{B}\langle \psi_{l,i} |=\boldsymbol{0}.$ So in this setting, we have 
		\begin{equation}\label{eq:SetSimple}			   \mathcal{D}_{A|B}(\mathcal{S})=\text{span}_\mathbb{C}\{|\xi_k\rangle_B\langle \xi_l| \mid 1\leq k\neq l\leq N,  \langle \theta_l|\theta_k\rangle_A\neq 0\}.
		\end{equation}

		\section{constructing small locally stable set of product states}\label{sec:smallset}
		Fix a multipartite systems $ \otimes_{i=1}^n\mathcal{H}_{A_i}$ where the dimension of each subsystem is at least 2, i.e.,  $\mathrm{dim}_{\mathbb{C}}(\mathcal{H}_{A_i})=d_i\geq 2$. It is interesting to consider the minimal set of product states in $\otimes_{i=1}^n\mathcal{H}_{A_i}$ that is locally stable. We define the  minimum size of locally stable set of  product states  in  $\otimes_{i=1}^n\mathcal{H}_{A_i}$ as $P(d_1,d_2,\cdots,d_n).$   First, we give an estimation on its lower bound.

		\begin{theorem}[On Lower Bound]\label{thm:ProductBasis_Bounds}
			Let $\cS$ be an orthogonal   set of  pure product states in  $\otimes_{i=1}^n\mathcal{H}_{A_i}$   whose local dimension  is  $\text{dim}_\mathbb{C}(\mathcal{H}_{A_i})=d_{i}$. 
			If  the set  $\cS$  is  locally  stable, then  
			$$(|\cS|-1) |\cS|\geq \sum_{i=1}^n(d_i^2-1).$$
			Particular, we   have $P(d_1,d_2,\cdots,d_n)\geq \frac{-1+\sqrt{(1+4D)}}{2}$ where $D=\displaystyle\sum_{i=1}^n(d_i^2-1).$
			
		\end{theorem}
		\begin{proof}
			Set $\mathcal{S}=\{\otimes_{i=1}^n |\psi^{A_i}_j\rangle \}_{j=1}^l$ where $l=|\mathcal{S}|$. 
			For each $i\in[n]$, we define
			\begin{equation*}
				\begin{array}{rl}
					C_{\hat{A}_i|A_i}(\mathcal{S}): =&\{(j,k)\in[l]\times [l]\ \mid  \langle \psi^{\hat{A}_i}_k |\psi^{\hat{A}_i}_j\rangle \neq 0,   j\neq k \},%\\
					%	C_{\hat{A}_i|A_i}(\mathcal{S}): =&\{(j,k)\in[l]\times [l]\ |  \langle \psi^{ {A}_i}_k |\psi^{ {A}_i}_j\rangle \neq 0,   j\neq k \},
				\end{array}				
			\end{equation*}		
			here  $|\psi^{\hat{A}_i}_j\rangle:=\otimes_{r\in [n]\setminus\{i\}} |\psi_j^{A_r}\rangle.$ As $\mathcal{S}$ is an orthogonal set, for each $(j,k)\in[l]\times [l]$ and $j\neq k$, there exists some $r\in[n]$ such that  $\langle \psi^{ {A}_r}_k |\psi^{ {A}_r}_j\rangle = 0.$ For all $i\in [n] \setminus\{r\},$  $\langle \psi^{\hat{A}_i}_k |\psi^{\hat{A}_i}_j\rangle=0$ by  definition. Hence   $(j,k)\notin C_{\hat{A}_i|A_i}(\mathcal{S})$  for such $i.$ Hence, $(j,k)\in C_{\hat{A}_i|A_i}(\mathcal{S})$  only if $i=r$. So each $(j,k)$ belong to at most one of the   sets $C_{\hat{A}_i|A_i}(\mathcal{S}), {i\in[n]}$.  This implies that the sets $C_{\hat{A}_i|A_i}(\mathcal{S}), i\in [n]$  are pairwise disjoint. As  each $C_{\hat{A}_i|A_i}(\mathcal{S})$  is a subset  of $[l]\times [l] \setminus \{(r,r)|r\in [l]\}$, by the  disjointness of $C_{\hat{A}_i|A_i}(\mathcal{S})$ ($i\in [n]$), we have
			\begin{equation}\label{eq:bound1}  \sum_{i=1}^n |C_{\hat{A}_i|A_i}(\mathcal{S})| \leq (l-1)l.
			\end{equation}
			By Eq. \eqref{eq:SetSimple}, we have  
			\begin{equation*}
				\begin{array}{cc}
					\mathcal{D}_{\hat{A}_i|A_i}(\mathcal{S})=\mathrm{span}_\mathbb{C}(\{ |\psi^{{A}_i}_j\rangle  \langle \psi^{ {A}_i}_k| \  \big |\   (j,k)\in  	C_{\hat{A}_i|A_i}(\mathcal{S})  \}).%\\
					%	\mathcal{D}_{\hat{A}_i|A_i}(\mathcal{S})=\mathrm{span}_\mathbb{C}(\{ |\psi^{\hat{A}_i}_j\rangle  \langle \psi^{ \hat{A}_i}_k| \  \big |\   (j,k)\in  	C_{\hat{A}_i|A_i}(\mathcal{S})   \}).
				\end{array}
			\end{equation*}
			So we get $ |C_{\hat{A}_i|A_i}(\mathcal{S})| \geq  \mathrm{dim}_{\mathbb{C}}[\mathcal{D}_{\hat{A}_i|A_i}(\mathcal{S})].$ As $\mathcal{S}$ is locally stable, by its equivalent form as in Eq. \eqref{eq:dimeq}, we obtain  
			$  \mathrm{dim}_{\mathbb{C}}[\mathcal{D}_{\hat{A}_i|A_i}(\mathcal{S})]=d_i^2-1$ 
			for $i\in [n].$  Therefore, we have the following inequalities 
			$$  |C_{\hat{A}_i|A_i}(\mathcal{S})| \geq \mathrm{dim}_{\mathbb{C}}[\mathcal{D}_{\hat{A}_i|A_i}(\mathcal{S})]=d_i^2-1$$ 
			for $i\in [n].$ 
			Summing up these terms, we have
			\begin{equation}\label{eq:bound2}  \sum_{i=1}^n |C_{\hat{A}_i|A_i}(\mathcal{S})| \geq\sum_{i=1}^n (d_i^2-1).
			\end{equation}
			From inequalities \eqref{eq:bound1} and \eqref{eq:bound2}, we obtain 
			$$(|\cS|-1) |\cS|=(l-1)l\geq \sum_{i=1}^n |C_{\hat{A}_i|A_i}(\mathcal{S})|\geq  \sum_{i=1}^n(d_i^2-1).$$
			This completes the proof.				
		\end{proof}
		
		\vskip 10pt
		
		In Theorem \ref{thm:ProductBasis_Bounds}, set $n=3, d_1=d_2=d_3=2,$ we have $(P(2,2,2)-1)P(2,2,2)\geq 9$.  So $P(2,2,2)\geq 4.$ On the other hand,  one can easily check that  the following set of four states which forms a UPBs \cite{DiVincenzo03}
		$$
		|000\rangle, |+-1\rangle, |1+-\rangle, |-1+\rangle 
		$$ (where $|+\rangle=|0+1\rangle=|0\rangle+|1\rangle$ and $ |-\rangle=|0-1\rangle=|0\rangle-|1\rangle$)   is   locally stable in $\mathbb{C}^2\otimes \mathbb{C}^2\otimes \mathbb{C}^2$.  Therefore, $P(2,2,2)=4$. One can easily show that a set of three entangled states in $\mathbb{C}^2\otimes \mathbb{C}^2\otimes \mathbb{C}^2$ is enough to show the local stableness. For example, the set of the following three states
		$$|000\rangle+|111\rangle, |000\rangle-|111\rangle, |001\rangle+|010\rangle+|100\rangle$$
		is locally stable (this example can be generalized to $n$ qubit systems. So there is a set of three entangled states in $n$ qubit systems that is locally stable).    So it seems that the set with entanglement may help to increase the local stableness in some sense.

		Now we present a property of   locally stable sets which are helpful for deriving upper bound of $P(d_1,\cdots, d_n).$
		\begin{proposition}[On Upper Bound]\label{pro:Nonlocal_Tensor}
			Let $\mathcal{S}_1$   and $\mathcal{S}_2$ be  two 	locally stable sets of   states in  $\otimes_{i=1}^m\mathcal{H}_{A_i}$ and $\otimes_{j=1}^n\mathcal{H}_{B_j}$ respectively. Then there exists a locally stable set   with cardinality $|\mathcal{S}_1|+|\mathcal{S}_2|-1$   in the $(m+n)$-parties systems $(\otimes_{i=1}^m\mathcal{H}_{A_i})\otimes (\otimes_{j=1}^n\mathcal{H}_{B_j})$.  Particularly, we have the following inequality
			$$ P({\bm d}_A,{\bm d}_B)\leq P({\bm d}_A)+P({\bm d}_B)-1$$
			where ${\bm d}_A:=(d_{A_1},\cdots,d_{A_m})$ and ${\bm d}_B:=(d_{B_1},\cdots,d_{B_n})$.
		\end{proposition}
		\begin{proof}
			Let $\mathcal{S}_1$   and $\mathcal{S}_2$ be  two 	sets of states in  $\otimes_{i=1}^m\mathcal{H}_{A_i}$ and $\otimes_{j=1}^n\mathcal{H}_{B_j}$ respectively. Define $\mathcal{S}_1\otimes\mathcal{S}_2$ to be the following set of $(m+n)$-parties system $$\{|\psi\rangle\otimes |\phi\rangle\in (\otimes_{i=1}^m\mathcal{H}_{A_i})\otimes (\otimes_{j=1}^n\mathcal{H}_{B_j}) \ \big |\  |\psi\rangle\in \mathcal{S}_1, |\phi\rangle\in \mathcal{S}_2   \}.$$
			Fix any $|\psi_1\rangle\in\mathcal{S}_1$ and   $|\phi_1\rangle\in\mathcal{S}_2$, we define $(\mathcal{S}_1,|\psi_1\rangle)\bigsqcup (\mathcal{S}_2,|\phi_1\rangle)$  to be the   set
			$\big(\{|\psi_1\rangle\} \otimes \mathcal{S}_2\big) \cup\big(\mathcal{S}_1\otimes \{|\phi_1\rangle\}   \big).$
			One finds that the cardinality of $(\mathcal{S}_1,|\psi_1\rangle)\bigsqcup (\mathcal{S}_2,|\phi_1\rangle)$ is just $|\mathcal{S}_1|+ |\mathcal{S}_2|-1$ as 
			$$\big(\{|\psi_1\rangle\} \otimes \mathcal{S}_2\big) \cap\big(\mathcal{S}_1\otimes \{|\phi_1\rangle\}   \big)=\{|\psi_1\rangle\otimes|\phi_1\rangle\}.$$
			
			Now we prove that the set  $(\mathcal{S}_1,|\psi_1\rangle)\bigsqcup (\mathcal{S}_2,|\phi_1\rangle)$ is locally stable as a set of $(m+n)$-parities.  If $A_i$ starts the first orthogonality preserving local  measurement, then   the set of states $\{M_{A_i}\otimes \mathbb{I}_{\hat{A}_i}\otimes \mathbb{I}_{B} |\psi\rangle\otimes |\phi_1\rangle\  \big | \ | \psi\rangle\in \mathcal{S}_1\} $ are mutually orthogonal. This is equivalent to the orthogonality of the set
			$\{M_{A_i}\otimes \mathbb{I}_{\hat{A}_i}| \psi\rangle \  \big | \ |\psi\rangle\in \mathcal{S}_1\} .$  But the locally stableness of $\mathcal{S}_1$ implies that $M_{A_i}^\dagger M_{A_i}\propto \mathbb{I}_{A_i}$. If $B_j$ starts an orthogonality preserving local measurement first, we can also show that it is a trivial measurement similarly.

		\end{proof}
		
		\vskip 8pt
		
		We know that the UPBs are always locally indistinguishable but maybe locally reducible. For example, let $\mathcal{S}$ denote a UPB in $\mathbb{C}^3\otimes \mathbb{C}^3$. Then $\mathcal{S}$ and the   $7$ states $\{|i\rangle|j\rangle \big |\   \ i,j\in\mathbb{Z}_4,  i \text{ or } j =3\}$ form a UPB in 
		$\mathbb{C}^4\otimes \mathbb{C}^4.$
		One can easily check that this UPB is locally reducible. Therefore, it can not be locally stable. In the following, we find that some UPBs  that  are  locally stable.
		
		We have presented that the UPB 
		$\{
		|000\rangle, |+-1\rangle, |1+-\rangle, |-1+\rangle \} $ is   locally stable in $\mathbb{C}^2\otimes \mathbb{C}^2\otimes \mathbb{C}^2$.    More generally, for an $(2n-1)$-qubit  systems, we have a UPB with  $2n$ elements (More details, see Ref. \cite{DiVincenzo03}). For example, the first state $|\Psi_0\rangle$ is $|00\cdots 0\rangle$, the second state  $|\Psi_1\rangle$ is 
		$$|1\rangle|\psi_1\rangle|\psi_2\rangle\cdots|\psi_{n-1}\rangle|\psi_{n-1}^\perp\rangle\cdots |\psi_2^\perp \rangle|\psi_1^\perp\rangle$$ 
		here $|\psi_i\rangle$ and $|\psi_j\rangle$ ($i\neq j$) are chosen to be neither orthogonal nor identical and $|\psi_i\rangle$   is neither orthogonal nor identical to $|0\rangle$ for all $i$. The other states in
		the UPB are obtained by (cyclic) right shifting the second state step by step. For example, the third state  $|\Psi_2\rangle$ is
		$$|\psi_1^\perp\rangle|1\rangle|\psi_1\rangle|\psi_2\rangle\cdots|\psi_{n-1}\rangle|\psi_{n-1}^\perp\rangle\cdots |\psi_2^\perp \rangle.$$
		
		\begin{example}\label{thm: Local_Stable_UPB_Qubits}
			For any integer $n\geq 2$, the above set of UPB with $2n$ elements  is locally stable in $(\mathbb{C}^2)^{\otimes (2n-1)}.$
		\end{example}
		\noindent \emph{Proof.} Without loss of generality, we assume $A_1$ starts with the first measurement $M=M_m^\dagger M_m=\left(\begin{array}{cc}
			m_{00} & m_{01} \\
			m_{10}&  m_{11}
		\end{array}\right)$. Considering that
		$$\langle \Psi_0| M\otimes \mathbb{I}_{\hat{A}_1} |\Psi_1\rangle=0$$
		and $\langle 0 \cdots 0||\psi_1\rangle|\psi_2\rangle\cdots|\psi_{n-1}\rangle|\psi_{n-1}^\perp\rangle\cdots |\psi_2^\perp \rangle|\psi_1^\perp\rangle\neq 0$ as  	$|\psi_i\rangle$  is neither orthogonal nor identical to $|0\rangle$ for each $i$. Using this, one would deduce that
		$m_{01}=m_{10}=0.$ Considering the  third state and the last state, i.e.,
		$|\Psi_2\rangle$ and $|\Psi_{2n-1}\rangle$
		$$
		\begin{array}{rcl}
			|\Psi_2\rangle&=&|\psi_1^\perp\rangle|1\rangle|\psi_1\rangle|\psi_2\rangle\cdots|\psi_{n-1}\rangle|\psi_{n-1}^\perp\rangle\cdots |\psi_2^\perp \rangle,\\
			|\Psi_{2n-1}\rangle&=&|\psi_1\rangle|\psi_2\rangle\cdots|\psi_{n-1}\rangle|\psi_{n-1}^\perp\rangle\cdots |\psi_2^\perp \rangle|\psi_1^\perp\rangle|1\rangle.	 	
		\end{array}	 
		$$
		One notice that only the first part is orthogonal. With this one can show that
		$$\langle\psi_1|M|\psi_1^\perp\rangle=0.$$
		Let $|\psi_1\rangle =\alpha|0\rangle +\beta|1\rangle $, then $|\psi_1^\perp\rangle =\hat{\beta}|0\rangle -\hat{\alpha}|1\rangle $ where $\alpha,\beta\neq 0$. Substitute the two states to the above equation, one would obtain that $m_{00}=m_{11}.$
		
		Therefore, one conclude that $A_1$ can only start with a trivial measurement. By the symmetry, all the other parts can only start with a trivial measurement.
		\qed
		
		\begin{corollary}\label{cor:Local_Stable_boun d_qubits}
			Let $n\geq 5$ be an  integer.  Then there exists a locally stable set of product states with cardinality $n+1$  in $n$-qubit  systems $(\mathbb{C}^2)^{\otimes n}.$ Consequence, $$P(\underbrace{2,2,\cdots,2}_{n})\leq n+1.$$
		\end{corollary}
		
		One can check that any $n\geq 5$ can be written as the form $n=3x+5y$ for some $x,y\in\mathbb{N}.$ Applying Proposition \ref{pro:Nonlocal_Tensor} to  3-qubit $\mathrm{UPB}$ (with 4 elements) and 5-qubit $\mathrm{UPB}$   (with 6 elements)  by   $x$ and $y$ times respectively, one obtain a locally stable set with $$4\times x+6\times y-(x+y-1)=3x+5y+1=n+1$$ elements in $n$-qubit  systems. Note that the minimum size of UPB in   $n$-qubit system is strictly larger than $n+1$ when $n\geq 4$ is even. 
		
		In fact, we can also reduce this bound  by subsets of   UPB.
		We  rewrite the last $N:=2n-1$ states of the UPB constructed above as follows:
		\begin{equation}\label{eq: qubits_UPB}
			\begin{array}{l}
				|\Psi_1\rangle=	|1\rangle|\psi_1\rangle|\psi_2\rangle\cdots|\psi_{n-1}\rangle|\psi_{n-1}^\perp\rangle\cdots |\psi_2^\perp \rangle|\psi_1^\perp\rangle,\\
				|\Psi_2\rangle=	|\psi_1^\perp\rangle|1\rangle|\psi_1\rangle|\psi_2\rangle\cdots|\psi_{n-1}\rangle|\psi_{n-1}^\perp\rangle\cdots |\psi_2^\perp \rangle,\\
				|\Psi_3\rangle=|\psi_2^\perp \rangle|\psi_1^\perp\rangle|1\rangle|\psi_1\rangle|\psi_2\rangle\cdots|\psi_{n-1}\rangle|\psi_{n-1}^\perp\rangle\cdots,\\
				\vdots\\
				
				|\Psi_N\rangle=	|\psi_1\rangle|\psi_2\rangle\cdots|\psi_{n-1}\rangle|\psi_{n-1}^\perp\rangle\cdots |\psi_2^\perp \rangle|\psi_1^\perp\rangle |1\rangle.\\
			\end{array}
		\end{equation}
		
		\noindent{\bf Observation.} These $2n-1$ states are pairwise orthogonal. So there are at least $\binom{2n-1}{2}$ pairs of orthogonal relation. On the other hand, each party has exactly $(n-1)$ pairs of orthogonal relations (i.e., $\langle  \psi_i|\psi_i^\perp\rangle=0, i=1,\cdots,n-1$) and there are $2n-1$ parties.
		As $\binom{2n-1}{2}=(2n-1)(n-1)$, the orthogonal relations of the entire states are   arising from  just one party and all the other parties are non-orthogonal. That is, if  $|\Psi_j\rangle=\otimes_{i=1}^N|\Psi_j^{A_i}\rangle,|\Psi_k\rangle=\otimes_{i=1}^N|\Psi_k^{A_i}\rangle$ where $j,k\in\{1,2,\cdots,N\}$ and $j\neq k$, then there is only one $i\in[N]$ such that $\langle \Psi_j^{A_i}|\Psi_k^{A_i} \rangle =0.$ Deleting $k$ states from the $(2n-1)$ states would removing at most $k$-pairs of orthogonal relations for each party. If we delete any $(n-3)$ states from the above set defined by Eq. \eqref{eq: qubits_UPB} (we arrive a set  with $(n+2)$ elements), there are at least two pairs of orthogonal relations for each party.

		\begin{example}\label{thm: multi_qubit_less_UPB}
			Let  $n\geq 3$ be an integer.  Any set $\mathcal{S}$  with  $(n+2)$   states from those defined by  Eq. \eqref{eq: qubits_UPB} is locally stable in ${(\mathbb{C}^2)}^{\otimes(2n-1)}$.
		\end{example}
		
		\noindent \emph{Proof.} Assume that the $A_\delta$-party start with an orthogonality preserving local measurement $\{\pi_m=M_m^\dagger M_m\}.$ As our observation above, there are two pairs of states  
		$$\begin{array}{ll}
			|\psi_k\rangle_{\delta}|\phi_k\rangle_{\hat{\delta}},& |\psi_k^\perp\rangle_{\delta}|\theta_k\rangle_{\hat{\delta}};\\[1mm]
			|\psi_l\rangle_{\delta}|\phi_l\rangle_{\hat{\delta}},& |\psi_l^\perp\rangle_{\delta}|\theta_l\rangle_{\hat{\delta}}
		\end{array}
		$$
		in  $\mathcal{S}$ such that ${}_{\hat{\delta}}\langle \phi_k|\theta_k\rangle_{\hat{\delta}}\neq 0,  {}_{\hat{\delta}}\langle \phi_l|\theta_l\rangle_{\hat{\delta}}\neq 0.$
		As $\pi_m$ is an orthogonality preserving map, we have the following relations
		\begin{equation}
			\begin{array}{l}
				{}_\delta\langle\psi_k|{}_{\hat{\delta}}\langle \phi_k| \pi_m\otimes \mathbb{I}_{\hat{A}_\delta} |\psi_k^\perp\rangle_{\delta}|\theta_k\rangle_{\hat{\delta}}=0,\\[1mm]
				{}_\delta\langle\psi_l|{}_{\hat{\delta}}\langle \phi_l| \pi_m\otimes \mathbb{I}_{\hat{A}_\delta} |\psi_l^\perp\rangle_{\delta}|\theta_l\rangle_{\hat{\delta}}=0.		
			\end{array}	
		\end{equation}
		As   ${}_\delta\langle\psi_k|\psi_k^\perp\rangle_\delta=0$ and  ${}_\delta\langle\psi_l|\psi_l^\perp\rangle_\delta=0$, by our observation, we have  ${}_{\hat{\delta}}\langle\phi_k|\theta_k^\perp\rangle_{\hat{\delta}}\neq 0$ and  ${}_{\hat{\delta}}\langle\phi_l|\theta_l^\perp\rangle_{\hat{\delta}}\neq 0.$ Therefore, we have the following relations
		\begin{equation}\label{eq: UPB_Less_Orthogonal_Relation}
			{}_\delta\langle\psi_k| \pi_m  |\psi_k^\perp\rangle_{\delta} =0,\ 	{}_\delta\langle\psi_l| \pi_m  |\psi_l^\perp\rangle_{\delta} =0.
		\end{equation}
		The first equality implies that $\pi_m$ is diagonal under the basis $\{ |\psi_k\rangle_\delta, |\psi_k^\perp\rangle_\delta\}$. Suppose that $\pi_m=x |\psi_k\rangle_{\delta}\langle \psi_k|+ y|\psi_k^\perp\rangle_{\delta}\langle \psi_k^\perp|$ and $|\psi_l\rangle_\delta=\alpha|\psi_k\rangle_{\delta}+\beta|\psi_k^\perp\rangle_{\delta}$ where $\alpha,\beta\neq 0$ by assumption. Substitute the expressions $|\psi_l\rangle_\delta$ and  $|\psi_l^\perp\rangle_\delta$ to the second relation of Eq. \eqref{eq: UPB_Less_Orthogonal_Relation}, we obtain that 
		$$ \alpha\beta x-\alpha\beta y=0.$$
		As $\alpha\beta\neq0$, we deduce $x=y$. Hence, the measurement  element $\pi_m\propto \mathbb{I}_2$. \qed
		\vskip 5pt

		\begin{corollary}\label{cor:Local_Stable_boun d_qubits_less}
			Let $n\geq 5$ be an integer.  Then there exist  a locally stable set of product states in $n$-qubit systems $(\mathbb{C}^2)^{\otimes n}$ with cardinality $\mathcal{N}(n)$  where 
			$\mathcal{N}(n)$ is defined as follows.
			Consequence, $$P(\underbrace{2,2,\cdots,2}_{n})\leq\mathcal{N}(n):=
			\begin{cases}
				\frac{n+1}{2}+2, &  \text{ odd } n\geq 5; \\[2mm]
				\frac{n}{2}+4, &  \text{ even } n\geq 10.
			\end{cases}. $$
		\end{corollary}
		
		\noindent{\bf \emph{Remark}:} In the   proof of Example \ref{thm: multi_qubit_less_UPB}, we only use the property that for each of the $(2n-1)$ party $A_\delta$, there are two pairs of orthogonal relations  
		$${}_\delta\langle\psi_k|\psi_k^\perp\rangle_\delta={}_\delta\langle\psi_l|\psi_l^\perp\rangle_\delta=0$$
		here $|\Psi_k\rangle,|\Psi_l\rangle\in\mathcal{S}$. Let $N:=2n-1$ (here we assume $N> 36$). In the appendix, we will show that there exists a subset $\mathcal{T}$ (whose size is $|\mathcal{T}|=3\lceil \sqrt{N}\rceil$  of states  defined by  Eq. \eqref{eq: qubits_UPB}   with this aforementioned property. Using this results and that from  Theorem \ref{thm:ProductBasis_Bounds}, we have
		$$\frac{-1+\sqrt{1+12N}}{2}\leq P(\underbrace{2,2,\cdots,2}_{N})\leq 3\lceil \sqrt{N}\rceil $$
		for $N=2n-1.$ Note that this upper bound is much smaller than that of Corollary \ref{cor:Local_Stable_boun d_qubits_less} when $N$ is large.

		For the two-qutrit systems, it is well known that there is a UPB with $5$ states  in $\mathbb{C}^3\otimes \mathbb{C}^3.$  In fact, one can show that it is locally stable. Therefore, by Theorem \ref{thm:ProductBasis_Bounds}, we could deduce $P(3,3)=5.$  Moreover, by   Theorem \ref{thm:ProductBasis_Bounds},  we obtain $P(3,3,3)\geq 6$. On the other hand, there is a UPB with $7$ elements in $\mathbb{C}^3\otimes \mathbb{C}^3\otimes\mathbb{C}^3$ (see Ref. \cite{DiVincenzo03} for more details). By using Matlab, we can check the following statement. 
		\begin{example}\label{ex:UPB333}
			Let $\mathbf{Sep}$ denote the UPB in $\mathbb{C}^3\otimes \mathbb{C}^3\otimes \mathbb{C}^3$ defined as  $|p_i\rangle=|u_i\rangle\otimes |v_i\rangle\otimes|w_i\rangle $ where
			$$|{u}_i\rangle:=\mathcal{N}(\cos \frac{2\pi i}{7}|0\rangle+\sin \frac{2\pi i}{7}|1\rangle+ h|2\rangle), \ i=0,1,\cdots,6 $$
			and $v_i=u_{2i \mod  7}$, $w_i=  v_{3i \mod  7}$ with $h=\sqrt{-\cos \frac{4\pi i}{7}}$ and $\mathcal{N}=1/\sqrt{1-\cos \frac{4\pi i}{7}}$. Then any six elements of the  set $\mathbf{Sep}$   is locally stable.
		\end{example}
		
		So $P(3,3,3)=6$.	Using the results $P(2,2)=5$ and $P(3,3,3)=6$,  by Proposition \ref{pro:Nonlocal_Tensor}, we can obtain an upper bound on $P(\underbrace{3,3,\cdots,3}_{n})$ as follows.
		
		\begin{corollary}\label{cor:Local_Stable_boun d_qutrits}
			Let $n\geq 2$ be an integer.  Then there exists a locally stable set of product states in $n$-qutrit  systems $(\mathbb{C}^3)^{\otimes n}$ with cardinality less than $\frac{5}{3} n+2$.  Consequence, $$P(\underbrace{3,3,\cdots,3}_{n})\leq \frac{5}{3} n+2.$$
		\end{corollary}

		One can check that any $n\geq 2$ can be written as the form $n=2x+3y$ for some $x,y\in\mathbb{N}.$ Applying Proposition \ref{pro:Nonlocal_Tensor} to  2-qutrit  $\mathrm{UPB}$ (with 5 elements) and subset of 3-qutrits $\mathrm{UPB}$  (with 6 elements) by   $x$ and $y$ times respectively, one obtain a locally stable set with $$5\times x+6\times y-(x+y-1)=4x+5y+1\leq \frac{5}{3} n+2$$ elements in $n$-qutrit   systems.

		Note that there is a trivial lower bound  $$f(d_1,\cdots,d_n):=1+\sum_{i=1}^n(d_i-1)$$ on the size of a UPB in $\bigotimes_{i=1}^n\mathbb{C}^{d_i}$ (see Ref.   \cite{DiVincenzo03}  for more details).  From Corollary \ref{cor:Local_Stable_boun d_qubits} and  \ref{cor:Local_Stable_boun d_qutrits}, we could deduce that 
		$$\begin{array}{c}
			P(\underbrace{2,2,\cdots,2}_{n})<f(\underbrace{2,2,\cdots,2}_{n}),\\
			P(\underbrace{3,3,\cdots,3}_{n})<f(\underbrace{3,3,\cdots,3}_{n})
		\end{array}
		$$
		for $n\geq 10.$ 
		
		Although the UPBs may not be locally stable in general, we conjecture that those UPBs (even  some of it subsets) with minimum size are locally stable.

		\vskip 30pt	
		
		\section{Conclusion and Discussion}\label{sec:Conclusion} In this paper, we studied the locally stableness of a set of product states in multipartite quantum systems. We obtained a lower bound on the size of locally stable set of product states and pointed out that there may some gap between  the size of locally stable set of product states and of the entangled states. Then we obtained a property on the upper bound of smallest size of locally stable set of product states in some give multipartite systems. Via some examples in $n$ qubit systems and $n$ qutrit systems, we pointed out that not all UPBs are locally stable but those UPBs with smallest size might be. Surprisingly, a much small subset of UPB may also present this kind of nonlocality. 
		
		There are also some questions left to be considered. We conjecture that  those UPBs with smallest size for some given multipartite systems are always locally stable.  It is also interesting to determine the exact value of the smallest size of locally stable set of product states for a given multipartite systems. We hope that the study of local stable set of product states will enrich our understanding of the quantum nonlocality.

		\vskip 16pt

		\vspace{2.5ex}
		
		\noindent{\bf Acknowledgments}\, \,  Wang  is supported  by  National  Natural  Science  Foundation  of  China  (Grant No. 11901084),  the Basic and Applied Basic
		Research Funding Program of Guangdong Province (Grant No. 2019A1515111097),   
		the Research startup funds of DGUT (GC300501-103). 
		
		\vskip 10pt
		
		{
			$${\text{\textbf{APPENDIX  }}}$$
		}
		\begin{figure*}[ht]
			\centering
			\includegraphics[scale=0.36]{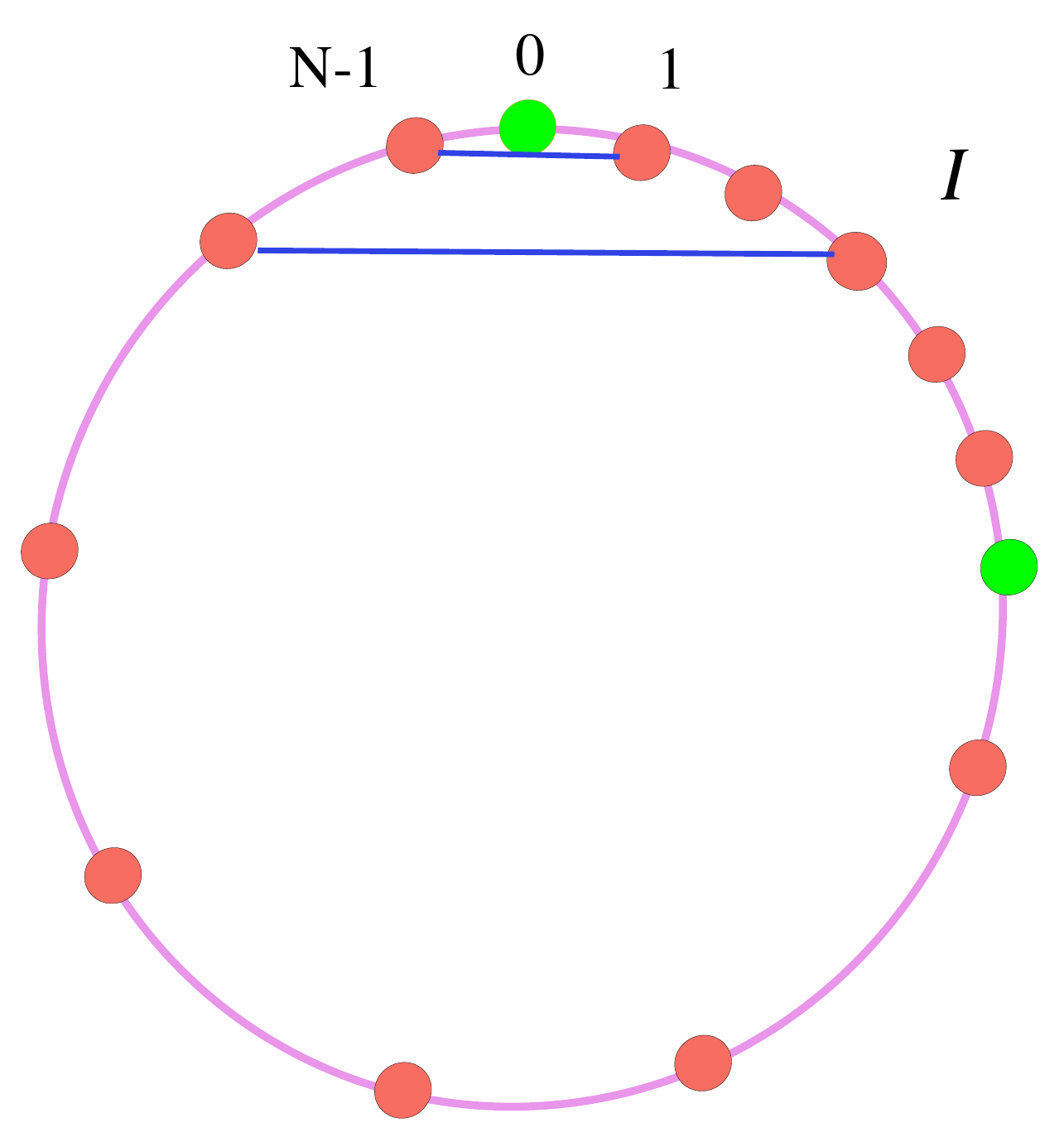}	\includegraphics[scale=0.41]{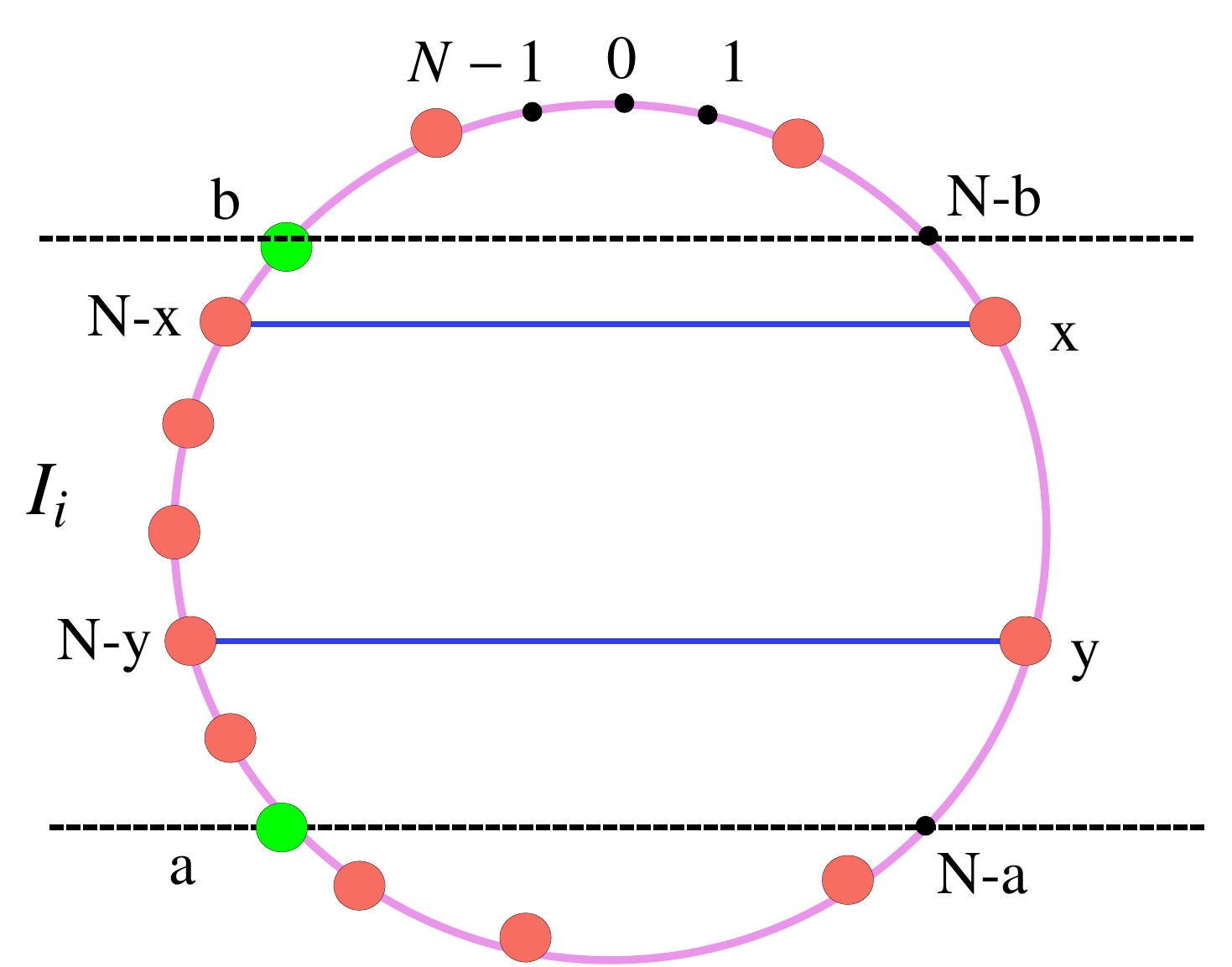}
			\includegraphics[scale=0.37]{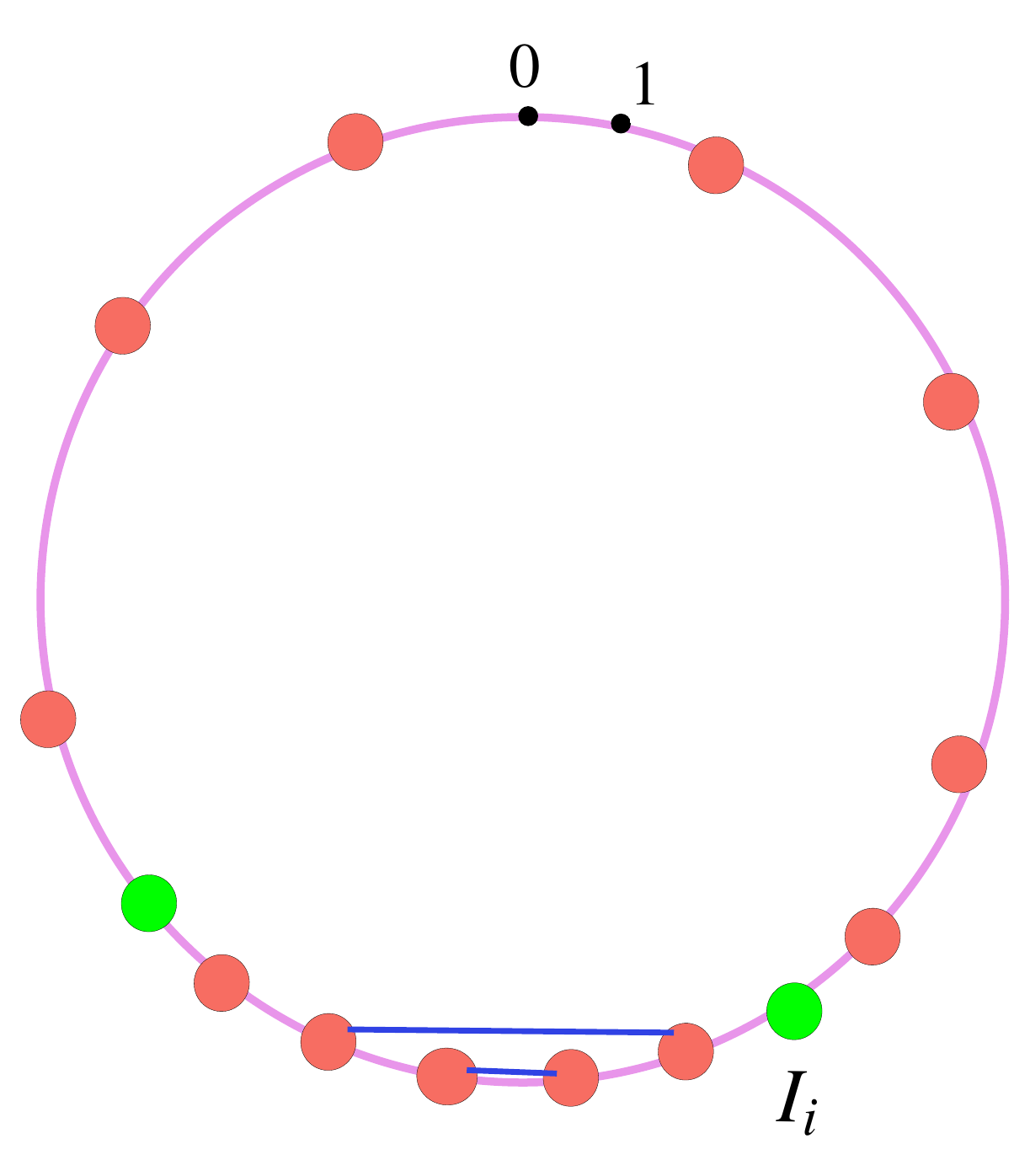}
			\caption{ An intuitive  view of the structure of $T_i$ and $I_i$ in the appendix. The dots between the two green dots (including the two green dots) are corresponding to the set $I_i$. The green and red dots   are corresponding to the set $T_i$. }\label{fig:structureT}
		\end{figure*}
		
		{\bf Construction of $\mathcal{T}$ in the remark of Corollary  \ref{cor:Local_Stable_boun d_qubits_less}:} Denote $\mathcal{T}\subseteq \{|\Psi_i\rangle\}_{i=1}^N$  the set we want to construct.
		We define a map  $\mathcal{F}$ from the cyclic group $\mathbb{Z}_N$ (where $N=2n-1$) to the set
		$$ \{|1\rangle, |\psi_1^\perp\rangle,|\psi_2^\perp\rangle,\cdots, |\psi_{n-1}^\perp\rangle,|\psi_{n-1},\rangle\cdots, |\psi_1 \rangle\}$$
		by setting 
		$$   \mathcal{F}(0)=|1\rangle, \ \  \mathcal{F}(i)=|\psi_i^\perp\rangle,\ \  \mathcal{F}(N-i)=|\psi_i\rangle$$
		for $i=1,\cdots,n-1$. Clearly, $\mathcal{F}$ is a bijection.
		\begin{equation}\label{eq: qubits_UPB_rename}
			\begin{array}{l}
				|\Psi_1\rangle=	\mathcal{F}(0)\mathcal{F}(N-1)\mathcal{F}(N-2)\cdots\mathcal{F}(2)\mathcal{F}(1),\\
				|\Psi_2\rangle=	\mathcal{F}(1)\mathcal{F}(0)\mathcal{F}(N-1)\cdots\mathcal{F}(3)\mathcal{F}(2),\\
				|\Psi_3\rangle=\mathcal{F}(2)\mathcal{F}(1)\mathcal{F}(0) \cdots\mathcal{F}(4)\mathcal{F}(3),\\
				\vdots\\
				
				|\Psi_{N}\rangle=	\mathcal{F}(N-1)\mathcal{F}(N-2)\mathcal{F}(N-3)\cdots\mathcal{F}(1)\mathcal{F}(0).\\
			\end{array}
		\end{equation}
		Note that the state $|\Psi_i\rangle$ can be completely determined by its first part. We assume that the set of the first part of $\mathcal{T}$ is 
		$$ P_1:=\{\mathcal{F}(t) \mid \ \ t\in T\},$$
		where $T\subseteq \mathbb{Z}_N$ and the cardinality of $T$ and $\mathcal{T}$ are equal, i.e., $|T|=|\mathcal{T}|.$ For each $i=2,\cdots,N,$ we define $$T_i:=T-(i-1):=\{t-(i-1) \mod N  \mid t\in T\}.$$ Then the $i$-th part of the states in $\mathcal{T}$ is just
		$$ P_i:=\{\mathcal{F}(t) \mid \ \ t\in T_i\}, \ i=2,\cdots,N.$$
		We want to find a small set $\mathcal{T}$ such that each part $P_i$ ($i=1,\cdots, n$) of $\mathcal{T}$ contains at least two pairs of orthogonal states.
		Let $i,j\in\mathbb{Z}_N$ then $\mathcal{F}(i)$ is orthogonal to $\mathcal{F}(j)$ if and only if $i,j\neq 0$ and $i+j\equiv 0 \mod N$.
		
		Let $\ell:=\lceil \sqrt{N}\rceil$. Define  $I=\{ 0,1,2,\cdots,2\ell+1\}$ and
		$$T=I\cup \{ k\ell \mid 3\leq k \leq  \ell -1\} \cup\{N-1\}.$$
		Then the cardinality of $T$ is $3\ell.$  For each $i=1,2,\cdots,N$, set 
		$$I_i:=I-(i-1):=\{t-(i-1) \mod N  \mid t\in I\}.$$
		We separate our discussion into the following cases.
		\begin{enumerate}
			\item [\rm(1)]  $i=0.$   As $0<N-(\ell-1)\ell\leq \ell$, we have $1,N-1,   (\ell-1)\ell, N-(\ell-1)\ell\in T$  with the property (see the left figure of Fig. \ref{fig:structureT}  )
			$$\begin{array}{l} 1+N-1=N\equiv 0 \mod N,\\    (\ell-1)\ell+N-(\ell-1)\ell=N\equiv 0\mod N.
			\end{array}
			$$
			\item [\rm(2)]  $1\leq i\leq 2\ell-1.$   In these cases, we have $1,N-1,  2 , N-2\in T_i$ with the property
			$$\begin{array}{l} 1+N-1=N\equiv 0 \mod N,\\    2+N-2=N\equiv 0\mod N.
			\end{array}
			$$
			\item [\rm(3)]  $ i= 2\ell.$   In these cases, we have $1,N-1, \ell , N-\ell\in T_i$ with the property
			$$\begin{array}{l} 1+N-1=N\equiv 0 \mod N,\\    \ell+N-\ell=N\equiv 0\mod N.
			\end{array}
			$$
			\item [\rm(4)]  $ i= 2\ell+1.$   In these cases, we have $\ell-1,N-\ell+1, 2\ell-1 , N-2\ell+1\in T_i$ with the property
			$$\begin{array}{l} \ell -1+N-\ell+1=N\equiv 0 \mod N,\\    2\ell-1+N-2\ell+1=N\equiv 0\mod N.
			\end{array}
			$$			    
			\item [\rm(5)]  $ 2\ell +2\leq i\leq n+1.$   In these cases, $0\notin I_i \subseteq[n,N-1]$.   	If the points in $\mathbb{Z}_N$ are looked as points in a circle, the points in   $I_i$ is a continuous interval of length  $2\ell+1$.  Suppose the minimum and maximum numbers in $I_i$ are $a$ and $b$ respectively. As the interval $[N-b,N-a]\subseteq [1,n-1]$ with length $2\ell +1$, it must contains at least two elements of $T_i$,  say, $x$ and $y$ (this is because the distance of any two elements in $ T_i $ are always less or equal to   $\ell$). There we have $x, N-x,y,N-y$ are nonzero elements in $T_i$ such that (see the middle figure of Fig. \ref{fig:structureT}  )
			$$\begin{array}{l} x+N-x=N\equiv 0 \mod N,\\    y+N-y=N\equiv 0\mod N.    \end{array}
			$$	
			\item  [\rm(6)] $ n+2 \leq i\leq N.$  This could separate into several cases which are similar to the above five cases. So we omit it here.
		\end{enumerate}
		To conclude, each $P_i$ (i.e., the $i$-th part of the set of the states in $\mathcal{T}$) contains two pairs of orthogonal states. \qed

	\end{document}